\documentclass{article}
\usepackage{geometry}

\title{Conditional Hardness of Earth Mover Distance}

\author{Dhruv Rohatgi \\ MIT \\ drohatgi@mit.edu}

\usepackage{enumitem}
\usepackage{amsmath}
\usepackage{graphicx}
\usepackage{amsthm}
\usepackage{amssymb}
\usepackage{hyperref} 
\usepackage{color}
\usepackage{calrsfs}
\usepackage{thmtools}
\usepackage{amssymb}
\usepackage{amsthm}
\usepackage{caption}
\usepackage{subcaption}
\usepackage{tikz}

\usetikzlibrary{shapes,arrows}
\newtheorem{theorem}{Theorem}[section]
\newtheorem{lemma}[theorem]{Lemma}

\newtheorem{proposition}[theorem]{Proposition} 

\theoremstyle{definition}

\theoremstyle{plain}

\newtheorem*{ovc}{Orthogonal Vectors Conjecture}
\newtheorem*{hsc}{Hitting Set Conjecture}

\newcommand{\norm}[1]{\left \lVert #1 \right \rVert}
\newcommand{\polylog}{\text{polylog}}
\newcommand{\poly}{\text{poly}}

\tikzstyle{block} = [rectangle, draw, fill=blue!10, 
    text width=6em, text centered, rounded corners, minimum height=2em]
\tikzstyle{line} = [draw, -latex']

\begin{document}

\maketitle

\begin{abstract}
The Earth Mover Distance (EMD) between two sets of points $A, B \subseteq \mathbb{R}^d$ with $|A| = |B|$ is the minimum total Euclidean distance of any perfect matching between $A$ and $B$. One of its generalizations is asymmetric EMD, which is the minimum total Euclidean distance of any matching of size $|A|$ between sets of points $A,B \subseteq \mathbb{R}^d$ with $|A| \leq |B|$. The problems of computing EMD and asymmetric EMD are well-studied and have many applications in computer science, some of which also ask for the EMD-optimal matching itself. Unfortunately, all known algorithms require at least quadratic time to compute EMD exactly. Approximation algorithms with nearly linear time complexity in $n$ are known (even for finding approximately optimal matchings), but suffer from exponential dependence on the dimension.

In this paper we show that significant improvements in exact and approximate algorithms for EMD would contradict conjectures in fine-grained complexity. In particular, we prove the following results:
\begin{itemize}
\item Under the Orthogonal Vectors Conjecture, there is some $c>0$ such that EMD in $\Omega(c^{\log^* n})$ dimensions cannot be computed in truly subquadratic time.
\item Under the Hitting Set Conjecture, for every $\delta>0$, no truly subquadratic time algorithm can find a $(1 + 1/n^\delta)$-approximate EMD matching in $\omega(\log n)$ dimensions.
\item Under the Hitting Set Conjecture, for every $\eta = 1/\omega(\log n)$, no truly subquadratic time algorithm can find a $(1 + \eta)$-approximate asymmetric EMD matching in $\omega(\log n)$ dimensions.
\end{itemize} 
\end{abstract}

\section{Introduction}

In the \emph{Earth Mover Distance (EMD) problem}, we are given two sets $A$ and $B$ each with $n$ vectors in $\mathbb{R}^d$, and want to find the minimum cost of any perfect matching between $A$ and $B$, where an edge between $a \in A$ and $b \in B$ has cost $\|a-b\|_2$.

In a harder variant of the problem (``EMD matching''), we want to actually \emph{find} a perfect matching with the optimal cost. This is a special case of the \emph{geometric transportation problem}, in which each vector of $A$ has a positive supply and each vector of $B$ has a positive demand, and the goal is to find an optimal ``transportation map'', i.e., match each unit of supply with a unit of demand while minimizing the total distance, summed over all units of supply. 

A more general variant of the EMD problem (with an analogous extension to arbitrary supplies/demands) allows for the possibility that $|A| < |B|$, and requires the map from $A$ to $B$ to be an injection.
We refer to this variant as the \textit{asymmetric} EMD problem.

Earth Mover Distance is a discrete analogue of the Monge-Kantorovich metric for probability measures, which has connections to various areas of mathematics \cite{Villani2003}. Furthermore, computing distance between probability measures is an important problem in machine learning \cite{Sandler2011, Mueller2015, Arjovsky2017, Flamary2016} and computer vision \cite{Rubner2000, Bonneel2011, Solomon2015}, to which Earth Mover Distance is often applied. To provide a few specific examples, computing geometric transportation cost has applications in image retrieval \cite{Rubner2000}, where asymmetric EMD allows the distance to deal with occlusions and clutter. In computer graphics, computing the actual transportation map is useful for interpolation between distributions, though the metric may be non-Euclidean \cite{Bonneel2011}. 

For the exact geometric transportation problem, the best known algorithm simply formulates the problem in terms of minimum cost flow, yielding a runtime of $O(n^{2.5} \cdot \polylog(U))$ where $U$ is the total supply (assuming that $d$ is subpolynomial in $n$) \cite{Lee2013, Lee2014}. Even for EMD, the best known algorithm follows directly from the general graph algorithms for maximum matching in $O(m \sqrt{n})$ time \cite{Hopcroft73}. 

The situation is better for approximation algorithms. There has been considerable work on both estimating the transportation cost \cite{Indyk2007, Andoni2014} and computing the actual map \cite{Sharathkumar2012, Agarwal2017, Altschuler2017} in time nearly linear in $n$ but exponential in dimension $d$. Most recently, it was shown \cite{Khesin2019} that there is an $O(n \epsilon^{-O(d)} \log(U)^{O(d)} \log^2 n)$ time algorithm which outputs a transportation map with cost at most $(1+\epsilon)$ times the optimum. This algorithm is very efficient when the dimension $d$ is constant or nearly constant, and when $\epsilon$ is not too small---say, constant or $O(1/\polylog(n))$. However, when $d = \omega(\log n)$, the algorithm is not guaranteed to find even a constant-factor approximation in quadratic time.

Despite considerable progress on improving the {\em algorithms} for geometric matching problems over the last two decades, little is known about {\em lower bounds} on their computational complexity. In particular, we do not have any evidence that a running time of the form $O(n \cdot \poly(d,\log n, 1/\epsilon))$ is not achievable. This is the question we address in this paper.

\subsection{Our Results}

In this paper we provide evidence that geometric transportation problems in high-dimensional spaces cannot be solved in (truly) subquadratic time. This applies to both exact and approximate variants of the problem, and even in the special case of unit supplies. In particular we show a conditional quadratic hardness for the exact EMD problem, as well as the approximate variant of EMD when the (approximately) optimal matching must be reported.

Our hardness results are based on two well-studied conjectures in fine-grained complexity: Orthogonal Vectors Conjecture and Hitting Set Conjecture (see \cite{VWilliams2018} for a comprehensive survey).

\subsubsection{Exact EMD and Orthogonal Vectors Conjecture} The \emph{Orthogonal Vectors (OV) problem} takes as input two sets $A, B \subseteq \{0,1\}^{d(n)}$ where $|A| = |B| = n$ and asks whether there are some vectors $a \in A$ and $b \in B$ such that $a \cdot b = 0$. The popular \emph{Orthogonal Vectors Conjecture} hypothesizes that in sufficiently large dimensions, the obvious quadratic time algorithm for OV is nearly optimal:

\begin{ovc}
Let $d(n) = \omega(\log n)$. For every constant $\epsilon > 0$, no randomized algorithm can solve $d(n)$-dimensional OV in $O(n^{2-\epsilon})$ time.
\end{ovc}

A plethora of problems have been shown to have nontrivial lower bounds under the Orthogonal Vectors Conjecture; often these lower bounds are essentially tight (e.g. \cite{Abboud2015, Abboud2016, Backurs2015, Bringmann2015, Williams2017}; see \cite{VWilliams2018} for a comprehensive survey). It is known that if the conjecture fails, then the Strong Exponential Time Hypothesis (SETH) fails as well \cite{Williams2005}, providing evidence for hardness of OV, and by extension of these problems to which OV can be reduced.

Our first result shows that EMD in ``nearly constant'' dimension is hard to compute exactly in truly subquadratic time, under the Orthogonal Vectors Conjecture:

\begin{theorem} \label{theorem:exactEMD}
There is a constant $c > 0$ under which the following holds. If there exists $\epsilon>0$ and $d(n) = \Omega(c^{\log^* n})$ such that EMD on $O(\log n)$-bit vectors in $d(n)$ dimensions can be computed in $O(n^{2-\epsilon})$ time, then the Orthogonal Vectors Conjecture is false.
\end{theorem}

Using techniques similar to those for the above theorem, we also address a question raised in \cite{Basu2018} about the complexity of the maximum/minimum weighted assignment problem when the weight matrix has low rank. The minimum weighted assignment problem is defined as follows: given an $n \times n$ weight matrix which determines a complete bipartite graph, find the cost of the minimum weight perfect matching. Motivated by the observation that the problem can be solved in $O(n \log n)$ time if the weight matrix is rank-$1$, it is asked whether there is an $O(nr^2 \log n)$ time algorithm for rank-$r$ matrices \cite{Basu2018}. We can answer this question in the negative, under the Orthogonal Vectors Conjecture. In fact, we can show something stronger (see Appendix~\ref{appendix:lowrank} for the proof):

\begin{theorem}\label{theorem:lowrank}
There is a constant $c > 0$ under which the following holds. If there exists $\epsilon>0$ and $r(n) = \Omega(c^{\log^* n})$ such that the minimum assignment problem with rank-$r$ weight matrices can be solved in $O(n^{2-\epsilon})$ time, then the Orthogonal Vectors Conjecture is false.
\end{theorem}

\subsubsection{Approximate EMD and the Hitting Set Conjecture} The second conjecture on which we base some of our results is hardness of the \emph{Hitting Set (HS) problem}. This problem, similar to OV, takes two sets of vectors $A, B \subseteq \{0,1\}^d$ as input, and asks whether there exists some $a \in A$ such that $a \cdot b \neq 0$ for every $b \in B$.

\begin{hsc}
Let $d(n) = \omega(\log n)$. For every constant $\epsilon > 0$, no randomized algorithm can solve $d(n)$-dimensional HS in $O(n^{2-\epsilon})$ time.
\end{hsc}

It is known that HS reduces to OV, but the reverse reduction is unknown, so the Hitting Set Conjecture is ``stronger'' than the Orthogonal Vectors Conjecture \cite{Abboud2016}. The Hitting Set Conjecture has been used to prove conditional hardness of the Radius problem in sparse graphs \cite{Abboud2016}. The utility of the Hitting Set problem in conditional hardness results comes from the difference between its ``$\exists \forall$'' logical structure and the ``$\exists \exists$'' logical structure of the Orthogonal Vectors problem, which makes it more natural for some types of problems. 

Under the Hitting Set Conjecture, we prove hardness of approximation for the EMD matching problem (in which we want to find the optimal or nearly-optimal matching). Simultaneously we obtain stronger hardness of approximation for asymmetric EMD matching.

\begin{theorem}\label{theorem:approxEMD}
For any $\delta > 0$ and $d(n) = \omega(\log n)$, if $(1 + 1/n^\delta)$-approximate EMD matching can be solved in $d(n)$ dimensions in truly subquadratic time, then the Hitting Set conjecture is false.
\end{theorem}

\begin{theorem}\label{theorem:asymEMD}
For any $d(n) = \omega(\log n)$ and $\eta = 1/\omega(\log n)$, if $(1 + \eta)$-approximate asymmetric EMD matching can be solved in $d(n)$ dimensions in truly subquadratic time, then the Hitting Set Conjecture is false.
\end{theorem}

Finally, motivated by the question of how hard Hitting Set really is, compared to Orthogonal Vectors, we generalize the result that Hitting Set reduces to Orthogonal Vectors by finding a set of approximation problems that lie between Orthogonal Vectors and Hitting Set in difficulty. For a positive integer function $k(n) \leq n/2$, we define the \textit{$(k, 2k)$-Find-OV problem}: given two sets $A, B \subseteq \{0,1\}^{d(n)}$ with $|A| = |B| = n$ and the guarantee that there exist at least $2k$ orthogonal pairs between $A$ and $B$, find $k$ pairs $\{(a_i, b_i)\}_{i=1}^k$ such that $a_i \cdot b_i = 0$ for every $i$.

We prove the following theorem in Appendix~\ref{section:k2kfindov}.

\begin{theorem}
Let $k(n) \leq n/2$. If $(k, 2k)$-Find-OV can be solved in truly subquadratic time, then the Hitting Set conjecture is false.
\end{theorem}

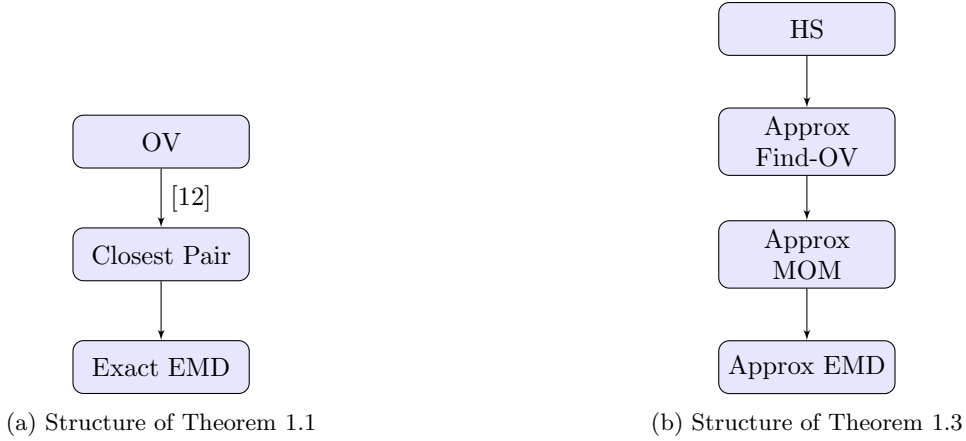
\begin{figure}
\centering
\begin{subfigure}[t]{0.48\textwidth}
\centering
\begin{tikzpicture}[ node distance = 1.5cm, auto]
	\node [block] (ov) {OV};
	\node [block, below of=ov] (closest) {Closest Pair};
	\node [block, below of=closest] (exact) {Exact EMD};

	\path [line] (ov) -- node {\cite{Chen2018}}(closest);
	\path [line] (closest) -- (exact);
\end{tikzpicture}

\caption{Structure of Theorem~\ref{theorem:exactEMD}}
\end{subfigure}\hfill
\begin{subfigure}[t]{0.48\textwidth}
\centering
\begin{tikzpicture}[ node distance = 1.5cm, auto]
	\node [block] (hs) {HS};
	\node [block, below of=hs] (afov) {Approx Find-OV};
	\node [block, below of=afov] (amom) {Approx MOM};
	\node [block, below of=amom] (aemd) {Approx EMD};

	\path [line] (hs) -- (afov);
	\path [line] (afov) -- (amom);
	\path [line] (amom) -- (aemd);
\end{tikzpicture}
\caption{Structure of Theorem~\ref{theorem:approxEMD}}
\end{subfigure}
\caption{Summary of reductions}
\label{figure:summary}
\end{figure}

See Figure~\ref{figure:summary} for an overview of the structure of our main results (Theorems~\ref{theorem:exactEMD} and \ref{theorem:approxEMD} respectively; the proof of Theorem~\ref{theorem:asymEMD} has the same structure as the latter). We provide the remaining definitions of the relevant problems in the next section.

\section{Preliminaries}

Before diving into the reductions, we formally define the remainder of the problems which we're studying. Each problem we study takes sets of vectors as input, so one parameter of a problem is the dimension $d$, which is a function of the input size $n$. That is, every function $d: \mathbb{N} \to \mathbb{N}$ defines a $d(n)$-dimensional EMD problem, and a $d(n)$-dimensional OV problem, and so forth. We gloss over this choice of $d$ in the subsequent definitions.

\subsection{Earth Mover Distance}

The \textit{Earth Mover Distance (EMD)} problem is defined as follows: given two sets $A, B \subseteq \mathbb{R}^{d(n)}$ with $|A|=|B|$, find $$\min_{\pi: A \to B} \sum_{a \in A} \norm{a - \pi(a)}_2$$ where $\pi$ is a bijection. We'll restrict our attention to the special cases where $A, B \subseteq \mathbb{Z}^{d(n)}$ with polynomially bounded entries (for hardness of exact EMD) and $A, B \subseteq \{0,1\}^{d(n)}$ (for hardness of approximate EMD).

We can define the \textit{asymmetric EMD} problem as above, except we relax the constraint $|A| = |B| = n$ to $|A| \leq |B| = n$, and require $\pi$ to be a injection rather than a bijection.

The \emph{EMD matching} problem is the variant of the EMD problem in which the desired output is the optimal matching $\pi$. Similarly we can define the \emph{asymmetric EMD matching} problem. An algorithm ``solves'' EMD matching (or its asymmetric variant) up to a certain additive or multiplicative factor if the cost of the bijection it outputs differs from the optimal cost by at most that additive or multiplicative factor.

\subsection{Variants of Orthogonal Vectors}

The reduction from Hitting Set to approximate EMD matching will go through the variants of OV defined next.

The \textit{Maximum Orthogonal Matching (MOM)} problem is defined as follows: given two sets $A, B \subseteq \{0,1\}^{d(n)}$, with $|A| \leq |B| = n$, find an injection $\pi: A \to B$ which maximizes $$|\{a \in A \mid a \cdot \pi(a) = 0\}|.$$

And the $\textit{Find-OV}$ problem is defined as follows: given two sets $A, B \subseteq \{0,1\}^{d(n)}$ with $|A| = |B| = n$, find the set $S \subseteq A$ of vectors $a \in A$ such that there exists some $b \in B$ with $a \cdot b = 0$. An algorithm solves Find-OV up to an additive error of $t$ if it returns a set $S' \subseteq S$ for which $|S'| \geq |S| - t$.

\subsection{Relevant prior work}

We will apply the following theorem from \cite{Chen2018} to our low-dimensional hardness result of exact EMD:

\begin{theorem}[\cite{Chen2018}]\label{thm:chen}
Assuming OVC, there is a constant $c>0$ such that Bichromatic $\ell_2$-Closest Pair in $c^{\log^* n}$ dimensions requires $n^{2 - o(1)}$ time, with vectors of $O(\log n)$ bit entries.
\end{theorem}

\section{Exact EMD in low dimensions}

To prove hardness of the exact EMD problem under the Orthogonal Vectors Conjecture, we reduce to the bichromatic closest pairs problem, and then apply Theorem~\ref{thm:chen} due to \cite{Chen2018}. The intuition for the reduction is as follows: given two sets $A$ and $B$ of $n$ vectors, we'd like to augment set $A$ with $n-1$ copies of a vector that is equidistant from all of $B$, and much closer to $B$ than $A$ is. Similarly, we'd like to augment set $B$ with $n-1$ copies of a vector that is equidistant from all of $A$, and much closer to $A$ than $B$ is. If this were possible, then the minimum cost matching between the augmented sets would only match one pair of the original sets: the desired closest pair.

Unfortunately, it is in general impossible to find a vector equidistant from $n$ vectors in $d \ll n$ dimensions. But this can be circumvented by embedding the vectors in a slightly higher-dimensional space, and adjusting coordinates in the ``free'' dimensions to ensure that an equidistant vector exists. So long as the free dimensions used to adjust set $A$ are disjoint from the free dimensions used to adjust set $B$, the inner products between $A$ and $B$ are unaffected, and the distances change in an accountable way.

Since we are working in the $\ell_2$ norm, we will need the following simple lemma which shows that any integer can be efficiently decomposed as a sum of a constant number of perfect squares.

\begin{lemma}\label{lemma:squaredecomposition}
For any $\rho > 0$ and any positive integer $m$, there is an $O(m^\rho)$ time algorithm to decompose $m$ as a sum of $O(\log 1/\rho)$ perfect squares. 
\end{lemma}

\begin{proof}
Here is the algorithm: repeatedly find the largest square which does not push the total above $m$, until the remainder does not exceed $O(m^{\rho/2})$. Then compute the minimal square decomposition for the remainder by dynamic programming.

The first, greedy phase takes $O(\polylog(m))$ time and finds $O(\log 1/\rho)$ perfect squares which sum to some $m'$ with $m - m^{\rho/2} \leq m' \leq m$. The second, dynamic programming phase takes $O(m^\rho)$ time (even naively). By Lagrange's four-square theorem, a decomposition of $m - m'$ into at most four perfect squares is found.
\end{proof}

Now we describe the main reduction of this section. We'll use a shorthand notation to define vectors more concisely: for example, $a^x b^y c^z$ refers to an $(x+y+z)$-dimensional vector with value $a$ in the first $x$ dimensions, $b$ in the next $y$ dimensions, and $c$ in the next $z$ dimensions.

\begin{theorem}\label{theorem:EMDCP}
Let $d = d(n) \leq n$ be a dimension, and let $k > 0$ be a constant. There is a constant $c = c(k)$ for which the following holds. Suppose that there is an algorithm which computes the $\ell_2$ earth mover distance between sets $A', B' \subseteq [1, n^{16k}]^{2d + 2c + 2}$ of size $n$ in $O(n^{2-\epsilon})$ time. Then bichromatic closest pair between sets $A, B \subseteq [1, n^k]^d$ of size $n$ can be computed in $O(n^{2-\epsilon})$ time as well.
\end{theorem}

\begin{proof}
Set $\rho = 1/(16k)$, and let $c = O(\log 1/\rho)$ be the constant in Lemma~\ref{lemma:squaredecomposition} for the number of perfect squares in a decomposition. Let $A$ and $B$ be two sets of vectors from $\{1, \dots, n^k\}^d$. Let $N = n^{16k}$. Our goal is to compute $$\min_{a \in A, b \in B} \norm{a - b}_2.$$ We can assume without loss of generality that $\norm{a}_2^2$ and $\norm{b}_2^2$ are odd for all $a \in A$ and $b \in B$: for instance, we can replace each vector $z = (z_1, \dots, z_d)$ by $(2z_1, \dots, 2z_d, 1)$.

We construct sets $A'$ and $B'$ of $(2d+2c+2)$-dimensional vectors as follows. Let $u = 0^d (10^c) 0^{c+1} 0^d$ (parentheses for clarity). Let $v = N^d 0^{c+1} (10^c) 0^d$. Add $n-1$ copies of $u$ to $B'$ and add $n-1$ copies of $v$ to $A'$. For each $a \in A$, add the following vector to $A'$, where we'll define vector $\text{adj}_a \in \mathbb{Z}^{c+1}$ later: $$a' = f(a) = 0^d (\text{adj}_a) 0^{c+1} a.$$ Similarly, for each $b \in B$, add the following vector to $B'$, where we'll define $\text{adj}_b \in \mathbb{Z}^{c+1}$ later: $$b' = g(b) = N^d 0^{c+1} (\text{adj}_b) b.$$

Now pick any $a \in A$. We'll construct $\text{adj}_a$ so that the following equalities are both satisfied: $$\norm{a' - u}^2_2 = n^{4k}d^2 = \norm{\text{adj}_a}^2_2.$$

Define the first element $\text{adj}_a(0) = (\norm{a}_2^2 + 1)/2$. Since $\norm{a}_2^2 \leq n^{2k}d$, we can then use Lemma~\ref{lemma:squaredecomposition} to find $c$ integers $\text{adj}_a(1),\dots,\text{adj}_a(c)$ so that $\norm{\text{adj}_a}^2_2 = n^{4k}d^2$. Furthermore,
\begin{align*}
\norm{a' - u}_2^2 
&= \norm{\text{adj}_a - 10^c}_2^2 + \norm{a}_2^2 \\
&= \norm{\text{adj}_a}_2^2 - 2 \cdot \text{adj}_a(0) + 1 + \norm{a}_2^2 \\
&= n^{4k}d^2.
\end{align*}

For each $b \in B$, we can similarly construct $\text{adj}_b$ so that $\norm{b' - v}_2^2 = \norm{\text{adj}_b}_2^2 = n^{4k}d^2.$

We claim that $$\text{EMD}(A', B') = 2(n-1)n^{2k}d + \min_{a \in A, b \in B}\sqrt{N^2d + 2n^{4k}d^2 + \norm{a-b}_2^2}.$$

To prove this claim, notice that $\norm{u-v}_2 \geq N\sqrt{d}$ and $\norm{a' - b'}_2 \geq N\sqrt{d}$ for every $a' \in A' \setminus \{v\}$ and $b' \in B' \setminus \{u\}$, whereas $\norm{a' - u}_2 \ll N\sqrt{d}/n$ and $\norm{b' - v} \ll N\sqrt{d}/n$. This means that the optimal matching between $A'$ and $B'$ will minimize the number of $(u,v)$ and $(a',b')$ edges. Hence, exactly one element of $A' \setminus \{v\}$ will be matched to an element in $B' \setminus \{u\}$. So if $M$ denotes this optimal matching, and $x' = f(x) \in A'$ is matched with $y' = g(y) \in B'$, then the cost of $M$ is
\begin{align*}
\text{cost}(M) 
&= \left(\sum_{a' \in A' \setminus \{v, x'\}} \norm{a' - u}_2 + \sum_{b' \in B' \setminus \{u, y'\}} \norm{b' - v}_2\right) + \norm{x' - y'}_2 \\
&= 2(n-1)n^{2k}d + \sqrt{N^2d + \norm{\text{adj}_{x}}_2^2 + \norm{\text{adj}_{y}}_2^2 + \norm{x - y}_2^2} \\
&= 2(n-1)n^{2k}d + \sqrt{N^2d + 2n^{4k}d^2 + \norm{x-y}_2^2}.
\end{align*}
The claim follows. So the algorithm is simply: run the EMD algorithm on $(A',B')$ and use the computed matching cost to find the closest pair distance, according to the above formula.

The time complexity of constructing $A', B'$ is $O(n^{5/4}d^{1/8})$, dominated by computing a square decomposition for each vector. Since $A'$ and $B'$ are sets of $O(n)$ vectors in $\mathbb{Z}^{2d + 2c + 2}$ with entries bounded by $\max(N, n^{2k}d) \leq n^{16k}$, the EMD between $A'$ and $B'$ can be computed in $O(n^{2-\epsilon})$ time. Thus, the overall algorithm takes $O(n^{2-\epsilon})$ time.
\end{proof}

Theorem~\ref{theorem:exactEMD} follows from the above reduction and Theorem~\ref{thm:chen}.

\section{Approximate EMD under the Hitting Set Conjecture}\label{section:approx}

In this section we prove hardness of approximation for the EMD matching problem when the approximately optimal matching must be reported. Note that the techniques from the previous section do not immediately generalize to this scenario, since the reduction in Theorem~\ref{theorem:EMDCP} is not approximation-preserving. A multiplicative error of $1 + \epsilon$ in the EMD algorithm would induce an additive error of $\tilde{O}(\epsilon n^{16k})$ in the closest pair algorithm, due to the large integers constructed in the reduction. A bucketing scheme, to ensure that the diameter of the input point set is within a constant factor of the closest pair, could eliminate the dependence on the values of the input coordinates, yielding a multiplicative error of only $1 + \tilde{O}(\epsilon n)$. 

However, $(1+\epsilon)$-approximate closest pair is only quadratically hard for $\epsilon = o(1)$ \cite{Rubinstein2018}; for any constant $\epsilon > 0$, there is a subquadratic $(1+\epsilon)$-approximation algorithm \cite{Indyk1998, Alman2016}. Thus, the above arguments would only yield $(1 + \tilde{O}(1/n))$-approximate hardness. Furthermore, the factor of $n$ loss intuitively feels intrinsic to the approach of reducing from closest pair, since the EMD is the sum of $n$ distances. Thus, a different approach seems necessary if we are to achieve hardness for $\epsilon = \omega(1/n)$.

Our method broadly consists of two steps. First, we show that EMD can encode orthogonality, by reducing approximate Maximum Orthogonal Matching (the problem of reporting a maximum matching in the implicit graph with an edge for each orthogonal pair) to approximate EMD matching. Second, we show that approximate Maximum Orthogonal matching can solve an instance $(A, B)$ of Hitting Set by finding an orthogonal pair $(a,b)$ for every $a \in A$ if possible, even if the set of orthogonal pairs does not constitute a matching.

We start by proving that asymmetric EMD matching reduces to EMD matching for the appropriate choices of error bounds. The reduction pads the smaller set of vectors $A$ with a vector that is equidistant from the opposite set $B$, so that its contribution to the earth mover distance can be accounted for. Of course, it is first necessary to transform the vectors so that an equidistant vector exists.

\begin{lemma}\label{lemma:symmetrize}
Suppose that $(1+\epsilon)$-approximate EMD matching in $D$ dimensions can be solved in $T(n,D)$ time. Then $(1+\epsilon)$-approximate asymmetric EMD matching in $d$ dimensions can be solved with an additional additive factor of $n\epsilon \sqrt{d}$ in $T(n,2d)$ time.
\end{lemma}

\begin{proof}
Let $A, B \subseteq \{0,1\}^{d}$ with $|A| \leq |B|$. Define sets $A', B' \subseteq \{0,1\}^{2d}$ by mapping $a \in A$ to the vector $$(a_1, \dots, a_d, 1 - a_1, \dots, 1 - a_d)$$ and similarly mapping $b \in B$ to $$(b_1,\dots,b_d, 1-b_1, \dots, 1-b_d).$$ Then add $|B|-|A|$ copies of the zero vector to $A'$.

Now $|A'| = |B'|$, so we can run the approximate EMD algorithm on $A'$ and $B'$ to find some bijection $\pi: A' \to B'$ such that $$\sum_{a' \in A'} \norm{a' - \pi(a')}_2 \leq (1 + \epsilon) EMD(A',B').$$

Each vector $b' \in B'$ has $\norm{b'}_2^2 = d$, so the distance from the zero vector to each match is exactly $\sqrt{d}$. And for any $a \in A$ and $b \in B$ which map to $a' \in A'$ and $b' \in B'$, $$\norm{a' - b'}_2^2 = 2\norm{a - b}_2^2.$$ Hence, the cost of $\pi$ is $$\sum_{a' \in A'} \norm{a' - \pi(a')}_2 = (|B| - |A|)\sqrt{d} + \sqrt{2} \cdot \sum_{a \in A} \norm{a - \pi(a)}_2$$ and the optimal cost is $$EMD(A', B') = (|B| - |A|)\sqrt{d} + \sqrt{2} \cdot EMD(A,B).$$ It follows that $$\sum_{a \in A} \norm{a - \pi(a)}_2 \leq \frac{\epsilon}{\sqrt{2}}(|B|-|A|)\sqrt{d} + (1+\epsilon) EMD(A, B),$$ which is the stated error bound.
\end{proof}

Next, we reduce approximate Maximum Orthogonal Matching to approximate asymmetric EMD matching. The general idea, given input sets $(A, B)$, is to deform $A$ and $B$ so that orthogonal pairs $(a,b)$ are mapped to pairs $(a'',b'')$ with distance $d_0$, and all other pairs are mapped to pairs with distance at least $d_1 > d_0$. Then add $|A|$ auxiliary vectors to $B$, each with distance exactly $d_1$ from all vectors in $A$. Thus, in an optimal matching, each vector of $A$ is either matched with an orthogonal vector at distance $d_0$, or some vector with distance exactly $d_1$. This introduces a nonlinearity, ensuring that in the additive matching cost, an orthogonal pair's contribution is not ``cancelled out'' by the contribution of a pair with dot product $2$, for instance. A similar trick was used by \cite{Backurs2015} in the context of edit distance, another ``additive'' metric.

The following simple lemma will be useful:

\begin{lemma}\label{lemma:negateproduct}
There are maps $\phi_1, \phi_2: \{0,1\}^d \to \{0,1\}^{3d}$ such that for any $a, b \in \{0,1\}^d$, $$\phi_1(a) \cdot \phi_2(b) = d - (a \cdot b).$$ Furthermore, the maps can be evaluated in $O(d)$ time.
\end{lemma}

\begin{proof}
Each dimension expands into three dimensions as follows:
$$a_i \mapsto (\phi_1(a)_{3i}, \phi_1(a)_{3i+1}, \phi_1(a)_{3i+2}) = (a_i, 1 - a_i, 1 - a_i)$$
$$b_i \mapsto (\phi_2(b)_{3i}, \phi_2(b)_{3i+1}, \phi_2(b)_{3i+2}) = (1 - b_i,b_i, 1 - b_i).$$
Then for each $i$, $$\sum_{j=3i}^{3i+2} \phi_1(a)_j \phi_2(b)_j = a_i(1-b_i) + (1-a_i)b_i + (1-a_i)(1-b_i) = 1 - a_ib_i.$$ Summing over $i=1,\dots,d$ we get $\phi_1(a) \cdot \phi_2(b) = d - (a \cdot b)$ as desired.
\end{proof}

\begin{lemma}\label{lemma:mom}
Suppose that $(1+\epsilon)$-approximate asymmetric EMD in $D$ dimensions can be solved with an additional additive factor of $n\epsilon \sqrt{D}$ in $T(n,D)$ time. Then the Maximum Orthogonal Matching problem in $d$ dimensions can be solved up to an additive factor of $O(n \epsilon d)$ in $T(2n, 12d+1)$ time.
\end{lemma}

\begin{proof}
Let $A, B \subseteq \{0,1\}^{d}$ with $|A| \leq |B| = n$. Define $A', B' \subseteq \{0,1\}^{3d}$ by $A' = \phi_1(A)$ and $B' = \phi_2(B)$, where $\phi_1, \phi_2$ are as defined in Lemma~\ref{lemma:negateproduct}.

Let $d' = 3d$ for convenience. Now we construct sets $A'', B'' \subseteq \{0,1\}^{4d'+1}$ as follows, starting from sets $A'$ and $B'$. We add $2d'$ dimensions to ensure that $\norm{a''}_2^2 = \norm{b''}_2^2 = d'$ for every $a'' \in A''$ and $b'' \in B''$ without changing the inner products. Add another $d'+1$ dimensions, extending each $a'' \in A''$ so that $a''_{3d'+1} = 1$ and $a''_i = 0$ otherwise; and extend each $b'' \in B''$ so that $b''_{3d'+2} = 1$ and $b''_i = 0$ otherwise. Finally augment $B''$ with $|A|$ copies of the vector $v \in \{0,1\}^{4d'+1}$ with $3d'$ zeros followed by $d'+1$ ones.

Notice that for every $a \in A$ and $b \in B$ corresponding to some $a'' \in A''$ and $b'' \in B''$, $$\norm{a'' - b''}_2^2 = \norm{a''}_2^2 + \norm{b''}_2^2 - 2 a'' \cdot b'' = 2(d'+1) - 2 a'' \cdot b'' = 2a \cdot b + 4d + 2,$$ and $$\norm{a'' - v}_2^2 = 2(d'+1) - 2 a'' \cdot v = 4d + 4.$$

Now we run the approximate asymmetric EMD matching algorithm on $A''$ and $B''$, yielding an injection $\pi: A'' \to B''$ such that $$\sum_{a'' \in A''} \norm{a'' - \pi(a'')}_2 \leq |B''|\epsilon \sqrt{4d'+1} + (1+\epsilon)EMD(A'', B'').$$ For each $a'' \in A''$, if $\norm{a'' - \pi(a'')}_2^2 > 4d+4$, then we can set $\pi(a'') = v$, preserving injectivity and decreasing the cost of the matching. Therefore every edge has cost either $\sqrt{4d+2}$ or $\sqrt{4d+4}$. In particular, if there are $m$ orthogonal pairs in the matching, the total cost is $$\sum_{a'' \in A''} \norm{a'' - \pi(a'')}_2 = m\sqrt{4d+2} + (|A|-m)\sqrt{4d+4}.$$ By the same argument as above, the minimum cost matching is obtained by maximizing the number of orthogonal pairs. If the maximum possible number of orthogonal pairs in a matching is $m_\text{OPT}$, then $$EMD(A'', B'') = m_\text{OPT}\sqrt{4d+2} + (|A|-m_\text{OPT})\sqrt{4d+4}.$$

Substituting these expressions into the approximation guarantee and solving, we get that $m \geq m_\text{OPT} - O(\epsilon n d)$ as desired.
\end{proof}

In the above lemma we assumed that we are given an algorithm for asymmetric EMD matching which has both a multiplicative error of $1+\epsilon$ and an additive error of $n\epsilon\sqrt{d}$, since this is the error introduced by the reduction to (symmetric) EMD. However, we are also interested in the hardness of $(1+\epsilon)$-approximate asymmetric EMD matching in its own right. Removing the additive error from the hypothesized algorithm in Lemma~\ref{lemma:mom} directly translates to an improved Maximum Orthogonal Matching algorithm, with an additive error of $O(\epsilon |A| d)$ instead of $O(\epsilon nd)$, where $n = |A| + |B|$:

\begin{lemma}\label{lemma:emdmom}
Suppose that there is an algorithm which solves $(1+\epsilon)$-approximate asymmetric EMD matching in $T(|A|+|B|,d)$ time, where the input is $A, B \subseteq \{0,1\}^d$. Then the Maximum Orthogonal Matching problem can be solved up to an additive error of $O(\epsilon |A| d)$ in $T(2n, 12d+1)$ time.
\end{lemma}

Now we could reduce OV to approximate Maximum Orthogonal Matching. The proof of the following theorem is given in Appendix~\ref{section:approxov} for completeness.

\begin{theorem}\label{thm:approxtoov}
Let $d = \omega(\log n)$. Under the Orthogonal Vectors Conjecture, for any $\epsilon > 0$ and $\delta \in (0,1)$, $(1+1/n^\delta)$-approximate EMD matching in $\{0,1\}^d$ cannot be solved in $O(n^{2\delta - \epsilon})$ time.
\end{theorem}

However, Theorem~\ref{thm:approxtoov} does not prove quadratic hardness for any approximation factor larger than $(1 + 1/n)$, and in fact breaks down completely for $(1+1/\sqrt{n})$-approximate EMD matching.

Instead, we reduce Hitting Set to approximate Maximum Orthogonal Matching, through approximate Find-OV. These two problems are structurally similar; the technical difficulty is that Find-OV may require finding many orthogonal pairs even when the largest orthogonal matching may be small, in which case applying the Maximum Orthogonal Matching algorithm would result in little progress. We resolve this with the following insight: if many vectors in set $A$ are orthogonal to at least one vector in set $B$ but there is not a large orthogonal matching, then some vector in set $B$ is orthogonal to \emph{many} vectors in $A$. But these vectors can be found efficiently by sampling.

In the proof of the following theorem we formalize the above idea.

\begin{theorem}\label{thm:findov}
Let $d = d(n)$ be a dimension. Suppose that the Maximum Orthogonal Matching problem can be solved up to an additive error of $E(|A|,|B|)$ in $O(n^{2-\epsilon}\poly(d))$ time, where the input is $A, B \subseteq \{0,1\}^d$. Then for any (sufficiently small) $\alpha > 0$ there is some $\gamma > 0$ such that Find-OV can be solved with high probability up to an additive error of $E(|A|, 2|B|^{1+\alpha})$ in $O(n^{2 - \gamma}\poly(d))$ time.
\end{theorem}

\begin{proof}
Let $A, B \subseteq \{0,1\}^d$ with $|A| = |B| = n$. Let $\alpha > 0$ be a constant we choose later. We may safely assume that $\alpha < 1$. Let the degree of a vector $a \in A$, denoted $d(a)$, be the number of $b \in B$ which are orthogonal to $a$. The algorithm for Find-OV consists of three steps:
\begin{enumerate}
\item For every $a \in A$, sample $n^{1 - \alpha/4}$ vectors from $B$ to get an estimate $\hat{d}(a)$ of $d(a)$. Mark and remove the vectors for which $\hat{d}(a) \geq n^{\alpha/2}$.

\item Next, for every $b \in B$, sample $n^{1 - \alpha/2}$ vectors from $A$ to get an estimate $\hat{d}(b)$ of $d(b)$. Let $B_\text{large} \subseteq B$ be the set of vectors for which $\hat{d}(b) \geq n^{\alpha}$. For each $b \in B_\text{large}$, iterate over $A$ and mark and remove each $a \in A$ for which $a \cdot b = 0$. Now remove $B_\text{large}$ from $B$.

\item Run the Maximum Orthogonal Matching algorithm on the remaining set $A$, and the multiset consisting of $2n^\alpha$ copies of each remaining $b \in B$. This produces a set of pairs $(a_i, b_i)$ where $a_i \cdot b_i = 0$. Output the union of $\{a_i\}_i$ and the set of all vectors marked and removed from $A$ in the previous steps.
\end{enumerate}

In the first step, a Chernoff bound shows that with high probability, every vector for which $d(a) \geq 2n^{\alpha/2}$ is marked and removed. Now summing over the remaining vectors, $$\sum_{a \in A} d(a) = \sum_{b \in B} d(b) \leq 2n^{1 + \alpha/2}.$$

In the second step, with high probability $B_\text{large}$ contains no $b \in B$ for which $d(b) \leq \frac{1}{2} n^\alpha$, by a Chernoff bound on each such $b \in B$. Therefore $|B_\text{large}| \leq 4n^{1 - \alpha/2}$. Furthermore, with high probability $B_\text{large}$ contains every $b \in B$ for which $d(b) \geq 2 n^\alpha$.

So after the first two steps, every remaining vector $b \in B$ has degree at most $2n^\alpha$. Suppose there are $t$ vectors $a \in A$ with positive degree, and $t'$ of these are found in the first two steps. Then by the degree bound, the remaining $t-t'$ vectors inject into $2n^\alpha$ copies of $B$. Therefore there is an orthogonal matching of size at least $t - t'$. By the approximation guarantee of the Maximum Orthogonal Matching algorithm, we find an orthogonal matching of size at least $t - t' - 2n^{(1+\alpha)\delta}$ in step 3. Overall, we find at least $t - 2n^{(1+\alpha)\delta}$ vectors with positive degree, which gives the desired approximation guarantee.

The time complexity is $O((n^{2 - \alpha/4} + n^{(2-\epsilon)(1+\alpha)})\poly(d))$. This is subquadratic in $n$ for sufficiently small $\alpha$.
\end{proof}

As the final step of the reduction, we show that approximate Find-OV can solve Hitting Set. Note that exact Find-OV obviously solves Hitting Set. It's also clear that Find-OV with an additive error of $n^{1-\epsilon}$ solves Hitting Set: simply run Find-OV, and then exhaustively check the remaining unpaired vectors of $A$---unless there are more than $n^{1-\epsilon}$ unpaired vectors, in which case there must be a hitting vector.

To reduce Hitting Set to Find-OV with additive error of $\Theta(n)$, the essential idea is simply to repeatedly run Find-OV on the remaining unpaired vectors. If the Find-OV algorithm has an additive error of $n/2$, then given an input $A, B$ with no hitting vector, the algorithm will find orthogonal pairs for at least $n/2$ vectors of $A$. Naively, we'd like to recurse on the remaining half of $A$. Unfortunately, the set $B$ cannot similarly be halved, so the error bound in the next step would not be halved. Thus, the algorithm might make no further progress.

The workaround is to duplicate every unpaired vector of $A$ before recursing. If $n/2$ orthogonal pairs are found but every vector of $A$ has been duplicated once, then matches are found for at least $n/4$ distinct vectors. This suffices to terminate the recursion in $O(\log n)$ steps.

\begin{theorem}\label{thm:hs}
Suppose that Find-OV in $d$ dimensions can be solved up to additive error of $n/2$ in $T(n,d)$ time. Then Hitting Set in $d$ dimensions can be solved in $O((T(n,d) + nd) \log n)$ time.
\end{theorem}

\begin{proof}
Let $A, B \subseteq \{0,1\}^d$ with $|A| = |B| = n$. Our hitting set algorithm consists of $t = \lceil \log n \rceil + 1$ phases. Initialize $R_1 = A$.

In phase $i \geq 1$, run Find-OV on $(2^{i-1} R_i, B)$, where $2^i R_i$ is the multiset with $2^i$ copies of each vector in $R_i$. Let $P \subseteq A$ be the output multiset and let $P'$ be the corresponding set (removing duplicates). Set $R_{i+1} = R_i \setminus P'$. If $|R_{i+1}| > n/2^i$, report failure (i.e. there is a hitting vector). Otherwise, proceed to the next phase. If phase $t$ is complete, report success (i.e. no hitting vector).

Suppose that the algorithm reports success. Then after phase $t$, we have $R_{t+1} \leq n/2^t < 1$. Then for every $a \in A$ there was some phase $i$ in which $a$ was removed from $R_i$, and therefore was orthogonal to some $b \in B$. So there is no hitting vector.

Suppose that the algorithm reports failure in phase $i$. Then $|R_i| \leq n/2^{i-1}$ and $|R_{i+1}| > n/2^i$, so $|P'| < n/2^i$. Therefore $|P| \leq 2^{i-1}|P'| < n/2$. By the Find-OV approximation guarantee, not every element of $R_i$ is orthogonal to an element of $B$. So there is a hitting vector.

The time complexity is dominated by $O(\log n)$ applications of Find-OV on inputs of size $O(n)$, along with $O(nd)$ extra processing in each phase. Thus, the time complexity is $O((T(n,d) + nd)\log n)$.
\end{proof}

The next theorem shows that hardness for approximate EMD matching (conditioned on the Hitting Set Conjecture) follows from chaining together the above reductions.

\begin{theorem}
If there are any $\epsilon, \delta > 0$ such that $(1+1/n^\delta)$-approximate EMD matching can be solved in $O(n^{2-\epsilon})$ time for some dimension $d = \omega(\log n)$, then the Hitting Set Conjecture is false.
\end{theorem}

\begin{proof}
Fix $d = \omega(\log n)$, and assume without loss of generality that $d(n)$ is polylogarithmic. Let $\epsilon, \delta > 0$ and suppose that $(1+1/n^\delta)$-approximate EMD matching can be solved in $O(n^{2-\epsilon})$ time. Then $(1+1/n^\delta)$-approximate asymmetric EMD can be solved with an additional additive error of $n^{1-\delta} \sqrt{d}$ with the same time complexity, by Lemma~\ref{lemma:symmetrize}. Hence, the Maximum Orthogonal Matching problem can be solved with an additive error of $n^{1-\delta} d$ in the same time, by Lemma~\ref{lemma:mom}. 

Applying Theorem~\ref{thm:findov} with parameter $\alpha = \delta$, we get a randomized algorithm for Find-OV with an additive error of $O(n^{1-\delta^2}d^{1+\delta})$ and time complexity $O(n^{2-\gamma})$ for some $\gamma > 0$. For sufficiently large $n$, the error is at most $n/2$. Thus, we can apply Theorem~\ref{thm:hs} to get a randomized algorithm for Hitting Set with time complexity $\tilde{O}(n^{2-\gamma})$, which contradicts the Hitting Set Conjecture.
\end{proof}

Furthermore, we obtain stronger hardness of approximation for asymmetric EMD matching:

\begin{theorem}
Let $d = \omega(\log n)$ and $\eta = 1/\omega(\log n)$. If there is a truly subquadratic $(1+\eta)$-approximation algorithm for asymmetric EMD matching in $d$ dimensions, then the Hitting Set Conjecture is false.
\end{theorem}

\begin{proof}
Fix $d' = \omega(\log n)$ and $\eta = 1/\omega(\log n)$ and $\epsilon > 0$. Suppose that there is an $O(n^{2-\epsilon})$ time algorithm which achieves a $(1+\eta)$ approximation for asymmetric EMD matching in $d'$ dimensions. Set $d = \min(d', \sqrt{(\log n)/\eta})$. Since $\mathbb{R}^d$ embeds isometrically in $\mathbb{R}^{d'}$, the algorithm also achieves a $(1+\eta)$ approximation for asymmetric EMD in $d$ dimensions.

By Lemma~\ref{lemma:emdmom}, the Maximum Orthogonal Matching problem can be solved up to an additive error of $O(\eta|A|d)$ in $O(d)$ dimensions and $O(n^{2-\epsilon})$ time. By Theorem~\ref{thm:findov} there is some $\gamma > 0$ such that Find-OV can be solved up to an additive error of $O(\eta nd)$ in $O(d)$ dimensions and $O(n^{2-\gamma})$ time. By choice of $d$ we have $\eta nd = o(n)$, so for sufficiently large $n$ the algorithm achieves additive error of at most $n/2$. Therefore by Theorem~\ref{thm:hs}, Hitting Set can be solved in $O(d)$ dimensions and $\tilde{O}(n^{2-\epsilon})$ time. Since $d = \omega(\log n)$, this contradicts the Hitting Set Conjecture.
\end{proof}

\paragraph{Acknowledgments.} I want to thank Piotr Indyk and Arturs Backurs for numerous helpful discussions and guidance. I am also grateful to an anonymous reviewer for pointing towards Theorem~\ref{theorem:lowrank} and its proof.

\bibliography{paper}

\appendix

\section{Hardness of Low-Rank Minimum Weighted Assignment}\label{appendix:lowrank}

The methods we used to prove hardness of exact EMD in low dimensions can be adapted to prove hardness of minimum weighted assignment with low-rank weight matrices, under the Orthogonal Vectors Conjecture. In particular, we show in the following theorem that bichromatic closest pair in $d$ dimensions can be reduced to minimum weighted assignment with a rank-$O(d)$ weight matrix. The reduction algorithm uses the same input transformation as Theorem~\ref{theorem:EMDCP}, and then solves minimum weighted assignment on the matrix $M$ with entries $M_{ij} = \norm{A'_i - B'_j}_2^2$, where $A'$ and $B'$ are the transformed input sets. The key is that $M$ has rank $O(d)$, and its minimum weight assignment encodes the squared closest pair distance of the input---just as the EMD of the transformed input in Theorem~\ref{theorem:EMDCP} encoded the closest pair distance of the input.

\begin{theorem}
Fix a dimension $d = d(n) \leq n$, and let $\epsilon > 0$. Suppose that there is an algorithm which solves minimum weighted assignment in $O(n^{2-\epsilon})$ time, if the weight matrix has rank at most $O(d)$. Then bichromatic closest pair in $d$ dimensions can be solved in $O(n^{2-\epsilon})$ time.
\end{theorem}

\begin{proof}
Let $A$ and $B$ be two sets of $n$ vectors in $d$ dimensions, with entries in $\{1, \dots, n^k\}$ for some constant $k>0$. Apply the transformation described in Theorem~\ref{theorem:EMDCP} to construct sets $A', B' \in \{0,\dots,n^{16k}\}^{2d+2c+2}$ where $c$ is as defined in the proof of the theorem. Define $$\text{SQEMD}(A', B') = \min_{\sigma: A' \to B'} \sum_{a' \in A'} \norm{a' - \sigma(a')}_2^2$$ where $\sigma$ ranges over all bijections from $A'$ to $B'$. Since $\norm{u - v}_2^2 \geq N^2d$ and $\norm{a' - b'}_2^2 \geq N^2 d$ for every $a' \in A' \setminus \{v\}$ and $b' \in B' \setminus \{u\}$, whereas $\norm{a' - u}_2^2 \ll N^2d/n$ and $\norm{b' - v}_2^2 \ll N^2 d/n$, the optimal matching $\sigma$ minimizes the number of $(u,v)$ and $(a',b')$ edges. In particular, exactly one element of $A' \setminus \{v\}$ is matched to an element of $B' \setminus \{u\}$. Thus, paralleling the proof of Theorem~\ref{theorem:EMDCP}, we get $$\text{SQEMD}(A', B') = 2(n-1)n^{4k}d^2 + \left(N^2 d + 2n^{4k}d^2 + \min_{a \in A, b \in B} \norm{a-b}_2^2\right).$$ 

Hence, to compute the bichromatic closest pair distance between $A$ and $B$, it suffices to compute $\text{SQEMD}(A', B')$. Representing $A'$ and $B'$ as $n \times (2d+2c+2)$ matrices, let $M$ be the $n \times n$ matrix defined by $M_{ij} = \norm{A'_i - B'_j}_2^2$. Then observing that $$M_{ij} = \sum_{k=1}^{2d+2c+2} (A'_{ik} - B'_{jk})^2 = \sum_{k=1}^{2d+2c+2} (A'_{ik})^2 + \sum_{k=1}^{2d+2c+2} (B'_{jk})^2 - 2\sum_{k=1}^{2d+2c+2} A'_{ik}B'_{jk},$$ we can write $M$ as the sum of $2d+2c+4$ rank-$1$ matrices, so $\text{rank}(M) \leq 2d+2c+4$. So by assumption, the minimum weight perfect matching in the complete bipartite graph determined by $M$ can be found in $O(n^{2-\epsilon} \poly(d))$ time. But the cost of the optimal matching is precisely $\text{SQEMD}(A', B')$.
\end{proof}

Applying Theorem~\ref{thm:chen} completes the proof of Theorem~\ref{theorem:lowrank}.

\section{Proof of Theorem~\ref{thm:approxtoov}}\label{section:approxov}

The theorem follows immediately from the reduction from Maximum Orthogonal Matching to EMD matching shown in section~\ref{section:approx}, and this next proposition.

\begin{proposition}
Suppose the Maximum Orthogonal Matching problem can be solved up to an additive factor of $n^\delta$ in $O(n^\gamma)$ time where $\delta < 1/2$. Then OV can be solved in $O(n^{\gamma/(1 - \delta)})$ time.
\end{proposition}

\begin{proof}
Let $A, B \subseteq \{0,1\}^d$ with $|A| = |B| = n$. We construct multisets $A'$ and $B'$ which consist of $2n^{\delta/(1-\delta)}$ copies of each $a \in A$, and $2n^{\delta/(1-\delta)}$ copies of each $b \in B$, respectively. We then run our approximate Maximum Orthogonal Matching algorithm on $A'$ and $B'$. If any orthogonal pair is found, we return it; otherwise we return that there is no orthogonal pair.

Since $|A'| = |B'| = 2n^{1/(1-\delta)}$, the time complexity of this algorithm is $O(n^{\gamma/(1 - \delta)})$. It is clear that if $A$ and $B$ have no orthogonal pair, then $A'$ and $B'$ have no orthogonal pair, so the algorithm correctly returns ``no pair''. 

Suppose that there are $a \in A$ and $b \in B$ with $a \cdot b = 0$ but the algorithm returns ``no pair''. Then the matching found by the algorithm had no orthogonal pairs. However, there is a matching consisting of $2n^{\delta/(1-\delta)}$ pairs. Since $|B'|^\delta < 2n^{\delta/(1-\delta)}$, this contradicts the approximation guarantee of the Maximum Orthogonal Matching algorithm.
\end{proof}

\section{Hardness of $(k, 2k)$-Find-OV}\label{section:k2kfindov}

The $(k, 2k)$-Find-OV problem provides some sense of the relative ``powers'' of the Orthogonal Vectors Conjecture and the Hitting Set Conjecture, as well as another example of how the Hitting Set Conjecture can be used to explain hardness of approximation problems. Reducing from OV, we get the following hardness result, and it is not clear how to make any improvement. Note that this proof extends to the $(1,2k)$-Find-OV problem, for which this lower bound is tight, due to a random sampling algorithm.

\begin{proposition}
Fix $\delta \in (0,1)$. Assuming OVC, any algorithm for $(n^\delta, 2n^\delta)$-Find-OV requires $\Omega(n^{2-\delta-o(1)})$ time.
\end{proposition}

\begin{proof}
Suppose that there exists an $O(n^{2-\delta-\epsilon})$ time algorithm \textsc{find} for $(n^\delta, 2n^\delta)$-Find-OV. Here is an algorithm for OV: given sets $A, B \subseteq \{0,1\}^d$ with $|A| = |B| = n$, duplicate each $a \in A$ and each $b \in B$ exactly $2n^{\delta/(2-\delta)}$ times. If the original number of orthogonal pairs was $r$, then the new number is $4rn^{2\delta/(2-\delta)}$. For $r \geq 1$, this exceeds $2(n \cdot 2n^{\delta/(2-\delta)})^\delta$, so applying \textsc{find} yields a positive number of orthogonal vectors if and only if $r > 0$. It's easy to check that the time complexity is subquadratic. 
\end{proof}

On the other hand, under the Hitting Set Conjecture, we can obtain quadratic hardness. When $k = n/2$, hardness follows from Theorem~\ref{thm:hs}, but it holds in greater generality. In particular, we provide a proof of conditional hardness for $k = \sqrt{n}$, and it extends naturally to any $k = n^\gamma$ for $\gamma \in (0,1)$. The proof takes inspiration from the reduction from Hitting Set to OV \cite{Abboud2016}, with a few extra twists.

\begin{theorem}
If the $(\sqrt{n}, 2\sqrt{n})$-Find-OV problem can be solved in $O(n^{2-\epsilon})$ time for some $\epsilon > 0$, then Hitting Set can be solved in $O(n^{2-\delta})$  time for some $\delta > 0$.
\end{theorem}

\begin{proof}
Let $\textsc{find}$ be the presupposed algorithm for $(\sqrt{n}, 2\sqrt{n})$-Find-OV. Set $\alpha = \epsilon/7$. Let $A, B \subseteq \{0,1\}^d$ with $|A| = |B| = n$. Without loss of generality, assume that no vector is all-zeroes. Here is an algorithm:
\begin{enumerate}
\item For each $a \in A$, randomly sample $n^{1-\alpha}$ vectors from B. If any of these is orthogonal to $a$, mark $a$ and remove it from $A$, replacing it with an all-ones vector.

\item Set $k = n^{1/3 - \alpha}$. Partition $A$ into sets $A_1,\dots,A_k$ of approximately equal size, and similarly partition $B$ into sets $B_1,\dots,B_k$. For each pair $(A_i, B_j)$:
\begin{enumerate}
\item Apply \textsc{find} to $(A_i, B_j)$.
\item If the output is not $\sqrt{n/k}$ orthogonal pairs, then continue to the next pair $(A_i, B_j)$.
\item Otherwise, suppose that the output is $\{(a_m, b_m)\}_{m=1}^{\sqrt{n/k}}$. For each vector $a \in \{a_m\}_{m=1}^{\sqrt{n/k}}$, mark $a$ and remove it from $A_i$ (and from $A$), replacing it with an all-ones vector.
\item Go to (a).
\end{enumerate}
\item If the number of unmarked input vectors exceeds $2n^{1 - 3\alpha/2}$, return ``NO'' and exit.
\item For each $a \in A$, if $a$ is not the all-ones vector, iterate over all $b \in B$, and mark $a$ if any $b \in B$ is orthogonal.
\item Return ``YES'' if every vector originally in $A$ is now marked, and ``NO'' otherwise.
\end{enumerate}
We claim that this algorithm solves Hitting Set in strongly subquadratic time. Correctness is relatively simple: a vector $a \in A$ is only marked by the above algorithm if some $b \in B$ is found for which $a \cdot b = 0$. Thus, if some $a \in A$ is a hitting vector for $B$, then it is never marked, so the algorithm returns ``NO''. 

Conversely, suppose that every $a \in A$ is orthogonal to some $b \in B$. Then the number of unmarked input vectors in Step 3 is at most the number of remaining orthogonal pairs. But each $(A_i, B_j)$ contains at most $2\sqrt{n/k}$ orthogonal pairs after Step 2 finishes, so the number of remaining orthogonal pairs in Step 3 is at most $k^2 (2\sqrt{n/k}) = 2n^{1-3\alpha/2}$. Thus, the algorithm continues to Step 4. Every $a \in A$ which has not been marked by the end of Step 2 is tested against every $b \in B$ in Step 4. Therefore every vector is marked, so the algorithm returns ``YES''.

Turning to time complexity, Step 1 takes $O(n^{2-\alpha})$ time. The complexity of Step 2 is dominated by the calls to \textsc{find}. For each pair $(A_i, B_j)$ there is at most one call to \textsc{find} for which the output is not $\sqrt{n/k}$ orthogonal pairs. Hence, there are $k^2 = n^{2/3 - 2\alpha}$ such ``failed'' calls. To bound the number of ``successful'' calls to \textsc{find}, for which the output \textit{is} $\sqrt{n/k}$ orthogonal pairs, note that after Step 1, with high probability each $a \in A$ is orthogonal to at most $n^{2\alpha}$ vectors $b \in B$, so the total number of orthogonal pairs is at most $n^{1+2\alpha}$. Each successful call eliminates $\sqrt{n/k} = n^{1/3 + \alpha/2}$ orthogonal pairs, so there are at most $n^{2/3 + 3\alpha/2}$ successful calls. This bound dominates the bound on failed calls. Each call takes time $O((n/k)^{2-\epsilon})$, so the time complexity of Step 2 is asymptotically $$n^{\left(\frac{2}{3}+\alpha\right)(2-\epsilon)} n^{\frac{2}{3} + \frac{3\alpha}{2}} = n^{2 - \frac{\epsilon}{6} - \frac{\epsilon^2}{7}}.$$

Step 3 takes negligible time. Finally, in Step 4, there are at most $2n^{1-3\alpha/2}$ vectors $a \in A$ which are not the all-ones vector (since each of these is unmarked), so the complexity is $O(n^{2 - 3\alpha/2})$.

Hence, the overall time complexity is bounded by $O(n^{2 - \epsilon/7})$.
\end{proof}

\end{document}